\documentclass[letterpaper, 10 pt, journal, twoside]{IEEEtran}
\IEEEoverridecommandlockouts
\usepackage{times}
\usepackage{setspace}
\usepackage{url}
\spacing{1}
\usepackage[utf8]{inputenc}
\usepackage[T1]{fontenc}
\usepackage{graphicx}		% For \includegraphics
\usepackage{wrapfig}
\usepackage[export]{adjustbox}
\usepackage{float}

\usepackage{amsmath} % assumes amsmath package installed
\usepackage{amssymb}  % assumes amsmath package installed
\usepackage{amsthm}
\usepackage{mathtools}
\usepackage[normalem]{ulem}
\usepackage{paralist}	% For enumerations with better titled numbers
\usepackage[space]{grffile} % Space in filenames
\usepackage{color}

\usepackage{enumitem}
\usepackage{bm}
\usepackage{cancel}
\usepackage{hhline}
\usepackage[c2 , nocomma]{optidef}
\usepackage{oubraces}
\usepackage{algorithmic}
\usepackage{textcomp}

\usepackage{booktabs}

\usepackage{amsfonts}
\usepackage{lipsum}

\usepackage{dsfont}

\usepackage{cite}
\usepackage{hyperref}

\newtheorem{theorem}{Theorem}

\theoremstyle{definition}
\newtheorem{definition}{Definition}
\theoremstyle{remark}
\newtheorem{remark}{Remark}
\theoremstyle{definition}
\newtheorem{assumption}{Assumption}
\theoremstyle{definition}

\newcommand{\R}{\mathbb{R}}
\newcommand{\C}{\mathcal{C}}

\definecolor{darkblue}{RGB}{0,0,102}
\definecolor{lightblue}{RGB}{77,77,148}

\definecolor{gold}{RGB}{234, 170, 0}
\definecolor{metallic_gold}{RGB}{139, 111, 78}

\newcommand{\mb}[1]{\mathbf{ #1 }}

\newcommand{\grad}{\nabla}

\DeclareMathOperator{\atan}{atan}

\DeclareMathOperator*{\argmin}{arg\,min}

\usepackage{xfrac}

\usepackage{dblfloatfix}

\usepackage{pdfpages}

\def\IEEEbibitemsep{0pt plus .5pt}

\usepackage{enumitem}

\usepackage{comment}

\newcommand*{\SetSuchThat}[1][]{} % reserve the name
\newcommand*{\MvertSets}{%
    \renewcommand*\SetSuchThat[1][]{%
        \mathclose{}%
        \nonscript\;##1\vert\penalty\relpenalty\nonscript\;%
        \mathopen{}%
    }%
}
\MvertSets % default
\DeclarePairedDelimiterX \Set [2] {\lbrace}{\rbrace}
    {\,#1\SetSuchThat[\delimsize]#2\,}

\renewcommand{\bf}{\mathbf{f}} 
\def\BibTeX{{\rm B\kern-.05em{\sc i\kern-.025em b}\kern-.08em
    T\kern-.1667em\lower.7ex\hbox{E}\kern-.125emX}}
\begin{document}

\title{\textbf{Safe Navigation under State Uncertainty: \\ Online Adaptation for Robust Control Barrier Functions}}

 \author{Ersin Da\c{s}$^{1}$, Rahal Nanayakkara$^{2}$, Xiao Tan$^{1}$, Ryan M. Bena$^{1}$, \\ Joel W. Burdick$^{1}$, Paulo Tabuada$^{2}$, and Aaron D. Ames$^{1}$
 % \thanks{Manuscript received: August, 26, 2025; Revised November, 21, 2025; Accepted December, 18, 2025.}%Use only for final RAL version
 % \thanks{This paper was recommended for publication by Editor Lucia Pallottino upon evaluation of the Associate Editor and Reviewers’ comments.% This work was supported by (organizations/grants which supported the work.) } %Use only for final RAL version
 \thanks{This work was supported in part by TII under project \#A6847 and in part by DARPA under the LINC Program. \textit{(Corresponding author: Ersin Da\c{s}.)}}
 \thanks{$^{1}$E. Da\c{s}, X. Tan, R. M. Bena, J. W. Burdick, and A. D. Ames are with the Department of Mechanical and Civil Engineering, California Institute of Technology, Pasadena, CA 91125, USA, \texttt{\{\!ersindas,\!xiaotan,\!ryanbena,\!jburdick,\!ames\!\}\!@caltech.edu}}
 \thanks{$^{2}$R. Nanayakkara and P. Tabuada are with the Electrical and Computer Engineering Department, University of California at Los Angeles, Los Angeles, CA 90095, USA, \texttt{rahaln@ucla.edu, tabuada@ee.ucla.edu}} 
 % \thanks{Digital Object Identifier (DOI) 10.1109/LRA.2026.3653366}
 }

\maketitle
%\thispagestyle{empty}         
% Removes the page number on the first page
% \pagestyle{plain}
% \pagestyle{empty}

% \begin{spacing}{0.93}
\begin{abstract}
Measurements and state estimates are often imperfect in control practice, posing challenges for safety-critical applications, where safety guarantees rely on accurate state information. In the presence of estimation errors, several prior robust control barrier function (R-CBF) formulations have imposed strict conditions on the input. These methods can be overly conservative and can introduce issues such as infeasibility, high control effort, etc. This work proposes a systematic method to improve R-CBFs, and demonstrates its advantages on a tracked vehicle that navigates among multiple obstacles. A primary contribution is a new optimization-based online parameter adaptation scheme that reduces the conservativeness of existing R-CBFs. In order to reduce the complexity of the parameter optimization, we merge several safety constraints into one unified numerical CBF via Poisson’s equation. We further address the dual relative degree issue that typically causes difficulty in vehicle tracking. Experimental trials demonstrate the overall performance improvement of our approach over existing formulations.
\end{abstract}


%\vskip -0.1 true in
%%%%%%%%%%%%%%%%%%%%%%%%%%%%%%%%%%%%%%%%%%%%%%%%%%%%%%%%%%%%
%%%%%%%%%%%%%%%%%%%%%%%%%%%%%%%%%%%%%%%%%%%%%%%%%%%%%%%%%%%%

\section{Introduction} 
\label{sec:intro}
%%%%%%%%%%%%%%%%%%%%%%%%%%%%%%%%%%%%%%%%%%%%%%%%%%%%%%%%%%%%
%%%%%%%%%%%%%%%%%%%%%%%%%%%%%%%%%%%%%%%%%%%%%%%%%%%%%%%%%%%%

\IEEEPARstart{T}{he} increasing use of autonomous robots in complex environments emphasizes the need for safety-critical control \cite{lindemann2023safe}. Control barrier functions (CBFs) \cite{ames2017control} provide conditions, by leveraging system dynamics, that render a prescribed safe set forward invariant and, in practice, lead to safety filters that minimally modify the nominal input only when necessary to maintain safety. However, model uncertainties, including external disturbances, time delays, and state estimation errors, can lead to safety violations in real-world control applications \cite{dacs2025robust}. For instance, modern mobile robots and uncrewed aerial vehicles often rely upon visual-inertial odometry (VIO) for localization and navigation. Under VIO, the uncertainty in the vehicle's location estimate grows between global positioning updates, and is an ever-present source of state estimation error that can render safety guarantees meaningless unless they are robustified against these persistent errors.

\begin{figure}[t]
    \centering    
    \vskip  1.7mm
    \includegraphics[width=\linewidth]{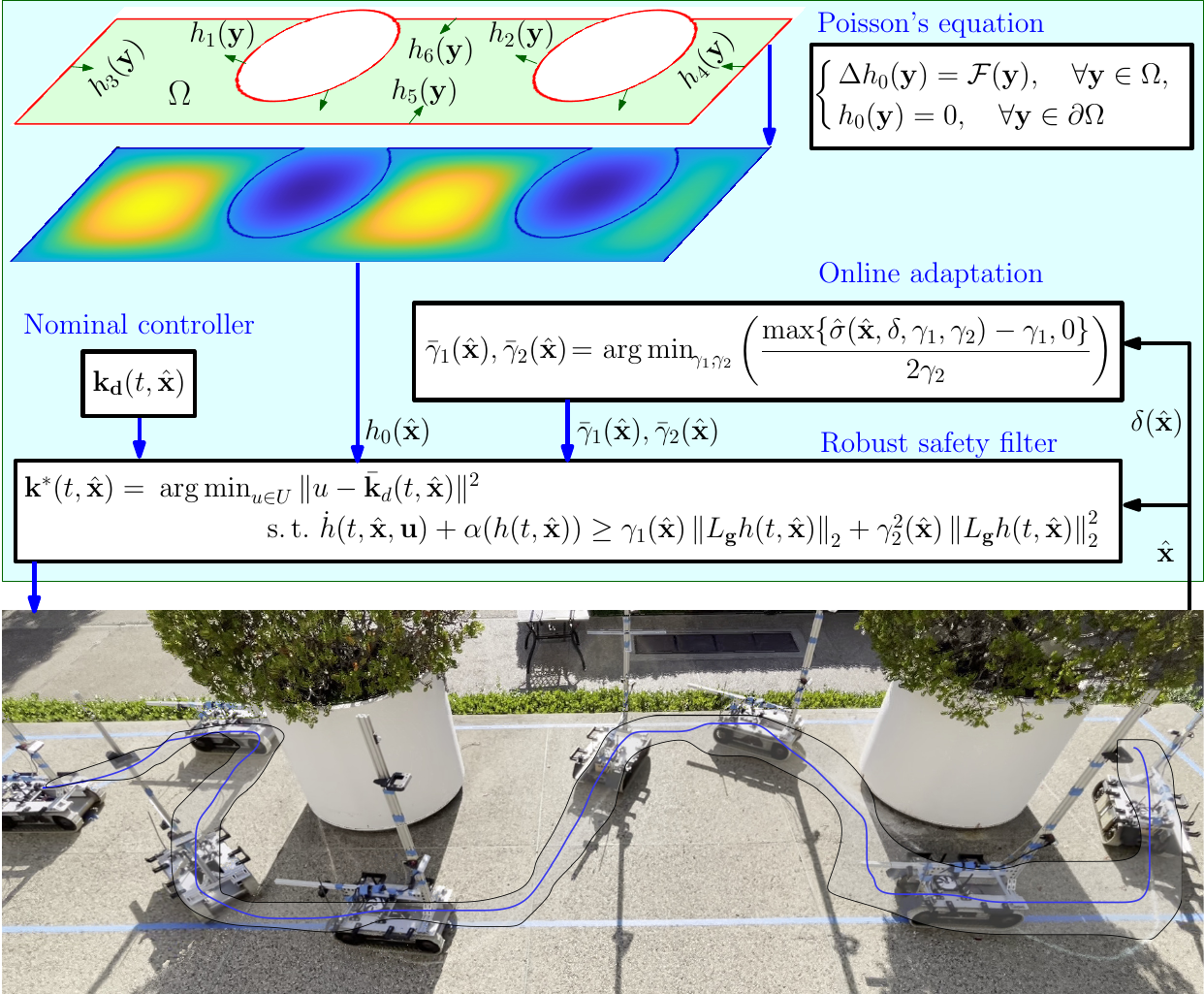}
    \caption{The proposed robust CBF (R-CBF) method in an outdoor safety-critical navigation scenario with a tracked mobile robot. Our R-CBF-QP-based safety filter maintains safety under state estimation uncertainty.}
    \label{fig: env}
    %\vskip - 4mm
\end{figure}

Early robust CBF methods addressed unmodeled dynamics through input‑to‑state safety margin \cite{alan2021safe}. Later formulations ensure safety by introducing adaptive robustness margins to reduce conservatism \cite{lopez2020robust, dacs2025robust}. One recent approach \cite{compton25a} learns a predictive robustness term that is updated online to tighten the CBF inequality whenever the full-order robot dynamics model deviates from its reduced-order model, assuming accurate state or output measurements. To enhance robustness in safety constraints, \cite{kim2024learning} presents a learning-based framework that enables online adaptation of class-$\mathcal{K}_{\infty}$ functions for systems under input constraints.

Beyond model uncertainty, robustness to measurement errors is notoriously challenging because CBF conditions are evaluated on the estimated state. This problem has been approached by using more conservative conditions \cite{dean2021guaranteeing, cosner2021measurement}, by considering specific types of state estimators \cite{wang2022observer, agrawal2022safe}, or by adopting stochastic approaches \cite{yang2023safe}. Although effective in specific control problems, each of these methods imposes specific requirements on the state estimator or uncertainty bounds, which may not be applicable to some real-world applications. To overcome these limitations, a concept for robust CBFs (R-CBFs) is introduced in \cite{rahal_cdc_2025} to ensure safety in the presence of state uncertainty, without requiring prior knowledge of the uncertainty magnitude. Unlike the approach proposed in \cite{agrawal2022safe, wang2022observer}, R-CBFs do not depend on specific classes of state observers. Nevertheless, R-CBFs account for estimation errors through conditions with constant robustness parameters, which may lead to overcompensation for uncertainty, ultimately reducing system performance in an undesirable manner.

In this paper, we present a safety-critical control framework that is robust to state estimation uncertainty, with a focus on practical implementation. We build upon the R-CBF formulation in \cite{rahal_cdc_2025} and propose an optimization-based procedure for the online tuning of the robustifying parameters in the R-CBF condition. This strategy enables the controller to adaptively guarantee robustness to state uncertainty by optimizing the R-CBF parameters using perturbations sampled locally about the current state estimate.

In order to account for multiple safety specifications---which commonly arise in real-world autonomy problems \cite{jang2024safe, long2021learning, dai2023safe, lindemann2018}---and reduce the computational complexity of the online parameter optimization, we consolidate them into a single CBF using Poisson safety functions \cite{bahati2025poisson}, \cite{bena2025geometry}. This also facilitates the generation of safe sets from perception data. Rather than composing safety constraints with approximations, the Poisson safety function provides a smooth scalar field by solving Poisson’s equation with Dirichlet boundary conditions defined by the unsafe and safe set boundaries.  

We demonstrate the effectiveness of our approach through hardware experiments on a tracked mobile robot, where state estimation is provided by VIO. We model the robot as a unicycle system.  Our objective is to provide collision safety even under the time-varying uncertainty introduced by the VIO system. The linear and angular velocity inputs to this vehicle affect the safety function at two different orders of differentiation, resulting in \textit{dual relative degree} behavior. Inspired by \cite{bahati2025control}, we first design a CBF for the single integrator sub-system, then we leverage a Lyapunov function to synthesize a CBF that has a relative degree of one with respect to both inputs for the nonlinear unicycle system. The synthesized CBF explicitly incorporates the effect of angular velocity into the safety condition, overcoming a limitation of conventional CBFs for the unicycle system, where angular velocity does not contribute to the satisfaction of the safety specifications.

We organize this paper as follows. Section~\ref{sec:pre} introduces preliminaries. Section~\ref{sec:theory} presents an online adaptation method for the robustifying parameters of R-CBFs, aimed at reducing conservativeness. Section~\ref{sec:imple} presents simulations and experiments, while Section~\ref{sec:conc} concludes the paper.

\section{Preliminaries} 
\label{sec:pre}
%\vspace{-10mm}
%%%%%%%%%%%%%%%%%%%%%%%%%%%%%%%%%%%%%%%%%%%%%%%%%%%%
%%%%%%%%%%%%%%%%%%%%%%%%%%%%%%%%%%%%%%%%%%%%%%%%%%%%

\subsection{Control Barrier Functions}
We consider a nonlinear control-affine system:
\begin{align} 
\label{eq:affine-dynamics}
    \dot{\mb x } = \mb{f}(\mb{x}) + \mb{g}(\mb{x}) \mb{u}, \
    \mb{x} \in \mathcal{X} \subseteq \R^n, \
    \mb{u} \in \mathcal{U} \subseteq \R^m,
\end{align}
where the drift dynamics ${\mb f \!:\! \mathcal{X} \!\to\! \R^n}$ and the actuation matrix ${ \mb g \!:\! \mathcal{X} \!\to\! \R^{n \times m}}$ are locally Lipschitz continuous functions. We assume the set of admissible control inputs $\mathcal{U}$ is an $m$-dimensional convex polytope. For a given time-varying feedback controller ${\mb k \!:\! \mathbb{R}_{\geq 0} \!\times\! \mathcal{X} \!\to\! \mathcal{U}}$, which is piecewise continuous in $t$ and locally Lipschitz continuous in $\mb x$, the closed-loop system, ${\mb f_{\rm cl} \!:\! \mathbb{R}_{\geq 0} \!\times\! \mathcal{X} \!\to\! \R^n}$:
\begin{equation}
\label{eq:clsystem1}
    \dot{\mb x} = \mb{f}(\mb x) + \mb{g}(\mb x) \mb{k} ( t, \mb x) \triangleq \mb f_{\rm cl}(t, \mb x),
\end{equation}
has a unique solution: ${  \mb {x}(t) \!=\! \mb x_0 \!+\! \int^{t}_{t_0} \mb f_{\rm cl}( \tau, \mb x(\tau)) d\tau ,~t \!>\! t_0.}$ for any initial condition ${\mb x(t_0) \!=\! \mb x_0 \!\in\! \mathcal{X}}$. In this paper, we assume that the system is forward complete.

We characterize safety through the notion of forward invariance. Let the time-varying safe set ${\C(t) }$ be the 0-superlevel set of a continuously differentiable function ${h\!:\! \mathbb{R}_{\geq 0} \!\times\! \mathcal{X} \!\to\! \mathbb{R}}$:
\begin{equation*}
%\label{eq:set}
    \C(t) \!\triangleq\! \left\{ \mb x \!\in\! \mathcal{X} \!\subseteq\! \mathbb{R}^n \!:\! h( t, \mb x) \!\geq\! 0 \right\}.
\end{equation*}
This set is said to be \textit{forward invariant} if, for any initial condition ${\mb x(t_0) \!\in\! \C(t_0)}$, the solution $ \mb {x}(t) $ stays in ${\C(t), ~\forall t \geq\! t_0}$. The closed-loop system \eqref{eq:clsystem1} is {\em safe} with respect to the \textit{safe set} $\C(t)$ if $\C(t)$ is forward invariant. To design controllers that guarantee such forward invariance properties, CBFs \cite{ames2017control} provide a systematic framework.
\begin{definition}[Control Barrier Function, \cite{ames2017control}]
\label{def:cbf}
Let ${\C(t) \!\subseteq\! \mathcal{X} }$ be the 0-superlevel set of a continuously differentiable function ${h\!:\! \mathbb{R}_{\geq 0} \!\times\! \mathcal{X} \!\to\! \mathbb{R}}$ with $0$ as a regular value. The function $h$ is a \textit{control barrier function} for \eqref{eq:affine-dynamics} on $\C(t)$ if there exists an extended class-$\mathcal{K}_{\infty}$ function\footnote{A continuous function ${\alpha \!:\! [0, a ) \!\to\! \mathbb{R}_{\geq 0}}$, where ${a \!>\! 0}$, belongs to class-${\mathcal{K}}$ (${\alpha \!\in\! \mathcal{K}}$) if it is strictly monotonically increasing and ${\alpha(0) \!=\! 0}$. A continuous function ${\alpha \!:\! \mathbb{R} \!\to\! \mathbb{R}}$ belongs to the set of extended class-$\mathcal{K}_\infty$ functions (${\alpha \!\in\! \mathcal{K}_{\infty}^{e}}$) if it is strictly monotonically increasing, ${\alpha(0) \!=\! 0}$, ${\lim_{r \!\to\! \infty} \alpha(r) \!=\! \infty}$ and ${\lim_{r \!\to\! -\infty} \alpha(r) \!=\! -\infty}$. A continuous function ${\beta \!:\! [0, a ) \!\times\! \mathbb{R}_{\geq 0} \!\to\! \mathbb{R}_{\geq 0}}$ belongs to class-${\mathcal{K L}}$ (${\beta \!\in\! \mathcal{KL}}$), if for every ${s \!\in\! \mathbb{R}_{\geq 0}}$, ${\beta(\cdot, s)}$ is a class-$\mathcal{K}$ function and for every ${r \!\in\! [0, a )}$, ${\beta(r, \cdot) }$ is decreasing and ${\lim_{s \!\to\! \infty} \beta(r, s) \!=\! 0.}$} ${\alpha \!\in\! \mathcal{K}_{\infty}^{e}}$ such that ${\forall \mb x \!\in\! \C(t)}$:
\begin{equation*}
% \label{eq:cbf}
   \sup_{\mb u \in \mathcal{U}} \big [ L_{\mb f} h(t, \mb x) + L_{\mb g} h(t, \mb x) \mb u + \partial_t h(t, \mb x) \big ] 
   \!\geq\! -\alpha (h(t, \mb x)),
\end{equation*}
where ${L_{\mb f} h \!=\! \frac{\partial h(t, \mb x)}{\partial x}\mb{f}(\mb{x})}$, ${L_{\mb g} h \!=\! \frac{\partial h(t, \mb x)}{\partial x}\mb{g}(\mb{x})}$, ${\partial_t h \!=\! \frac{\partial h(t, \mb x)}{\partial t}}$.
\end{definition}
We can utilize CBFs to design a {\em safety filter}.
In particular, consider a baseline \textit{nominal controller} ${\mathbf{k_d} \!:\! \mathbb{R}_{\geq 0} \!\times\! \mathcal{X} \!\to\! \mathcal{U}}$ that is locally Lipschitz continuous but potentially unsafe, along with a CBF ${h}$ with a corresponding ${\alpha \!\in\! \mathcal{K}_{\infty}^{e}}$ for system \eqref{eq:affine-dynamics}. We enforce safety constraints using the CBF-Quadratic Program (CBF-QP) \cite{ames2017control}:   
\begin{align*}
%{\label{CBF-QP}}
\begin{array}{l}
{\mathbf{k^*}(t, \mb x)= \ }
\displaystyle  \argmin_{{\mb u}  \in \mathcal{U}} \ \ \ {\|\mb u - \mb{k_d}( t, \mb x) \|^2}  \\ [1mm]
~~~~~~~~~~~~~~\textrm{s.t.} ~~~~~~~~~\dot{h}(t, \mb x, \mb u)  \geq - \alpha (h( t, \mb x)).
\end{array}
\end{align*}

%%%%%%%%%%%%%%%%%%%%%%%%%%%%%%%%%%%%%%%%%%%%%%%%%%%%
%%%%%%%%%%%%%%%%%%%%%%%%%%%%%%%%%%%%%%%%%%%%%%%%%%%%
\section{Robust Control Barrier Functions with Online Parameter Adaptation} 
\label{sec:theory}
%%%%%%%%%%%%%%%%%%%%%%%%%%%%%%%%%%%%%%%%%%%%%%%%%%%%
%%%%%%%%%%%%%%%%%%%%%%%%%%%%%%%%%%%%%%%%%%%%%%%%%%%%
We first define the concept of a time-varying robust CBF (R-CBF) to account for state uncertainty. Then, we introduce an optimization method for online adaptation of the R-CBF tuning parameters.

\subsection{Robust Control Barrier Functions}
In practice, feedback control systems rely on an estimate, ${\hat{\mb x}  \! \in \! \hat{\mathcal{X}} \!\subseteq\! \R^n}$ of the true state ${x \!\in\! \mathcal{X} \!\subseteq\! \R^n}$ for the calculation of the control input. This estimate is typically obtained through an observer or filter, and is subject to bounded estimation uncertainties satisfying ${\| {\mb x}  \!-\!  \hat{\mb x} \| \!\leq\! \delta({\mb x})}$, for some ${\delta \!:\! \mathcal{X} \!\to\! \mathbb{R}_{\geq 0}}$, where ${ \mb e \!\triangleq\!{\mb x}  \!-\!  \hat{\mb x}  }$ is the state estimation error.  
\begin{assumption}
\label{as:regular}
Every non-positive real number is a regular value of $h$, ${\forall t \!\geq\! 0}$, i.e., ${\nabla_{\mb x} h(t, \!{\mb x}) \!\neq\! 0}$ when ${h(t, \!{\mb x}) \!\leq\! 0}$, ${\forall t \!\geq\! 0}$.
\end{assumption}
\begin{assumption}
\label{as:compact}
Every non-positive super-level set of $h$ is compact and uniformly bounded, i.e., $\C_{\beta}(t) \!\triangleq\! \big\{{\mb x} \!\in\! {\mathcal{X}} \!:\! h(t, {\mb x}) \!\geq\! - \beta \big\}$ is compact for all $\beta \!\in\! \mathbb{R}_{\geq 0}, ~t \!\geq\! 0$, and $\C_{\beta}(t) \!\subseteq\! S$ for all $t \!\geq\! 0$, where $S$ is some compact set.
\end{assumption}
\begin{definition}
\label{def:stable_set}
A time-varying compact set ${\mathcal{S}(t) \!\subseteq\! \R^n}$ is asymptotically stable if there exists ${\rho \!\in\! \mathcal{K}\mathcal{L}}$ such that
\begin{equation*}
\|\mb x (t) \|_{\mathcal{S}(t)} \leq \rho(\|\mb x (t_0) \|_{\mathcal{S}(t_0)}, t \!-\! t_0), \ \forall t \geq t_0,
\end{equation*}
where ${\|\mb x(t) \|_{\mathcal{S}(t)} \!\triangleq\! \inf_{y \in \mathcal{S}(t)} \| \mb x(t) - y\|}$ is the distance between the point $\mb x$ and the set $\mathcal{S}$ at time $t$.
\end{definition}

The notion of a robust control barrier function \cite{rahal_cdc_2025} has been proposed to synthesize safety-critical controllers under state estimation uncertainty.
\begin{definition}[Robust Control Barrier Function, \cite{rahal_cdc_2025}]
\label{def:rcbf}
A continuously differentiable function ${h\!:\! \mathbb{R}_{\geq 0} \!\times\! {\mathcal{X}} \!\to\! \mathbb{R}}$, is a \textit{robust control barrier function} parameterized by state-dependent functions ${\gamma_{1}, \gamma_{2} \!:\! \hat{\mathcal{X}} \!\to\! \mathbb{R}_{> 0}}$ for system \eqref{eq:affine-dynamics} with respect to the 0-superlevel set of $h$, ${\C(t) \!\subseteq\! \mathcal{X} }$, if there exists ${\alpha \!\in\! \mathcal{K}_{\infty}^{e}}$ such that:
\begin{align}
\begin{split}
\label{eq:rob_cbf}
&\sup_{u \in \mathcal{U}} \big [ L_{\mb f} h(t, \hat{\mb x}) \!+\! L_{\mb g} h(t, \hat{\mb x}) \mb u \!+\! \partial_t h(t, \hat{\mb x}) \!+\! \alpha (h(t, \hat{\mb x})) \big ]  \\
&~~~~~\geq \gamma_{1}(\hat{\mb x}) \left\|L_{\mb g} h(t, \hat{\mb x})\right\|_{2} + \gamma_{2}^2 (\hat{\mb x}) \left\|L_{\mb g} h(t, \hat{\mb x})\right\|_{2}^{2},
\end{split}
\end{align}
for all ${\hat{\mb x}  \! \in \! \hat{\mathcal{X}}}$, where ${ \hat{\mb x} \!=\! {\mb x} \!-\! \mb e}$.
\end{definition}

Compared to the robust CBF formulation in \cite{rahal_cdc_2025}, which introduces a general robustness function, our R-CBF definition uses explicit state-dependent robustness terms $\gamma_1$ and $\gamma_2$. This structured choice is a specialization of the framework in \cite{rahal_cdc_2025} and is tailored to simplify the online adaptation procedure. Given an R-CBF ${h}$ and a corresponding ${\alpha}$, the pointwise set of all control values that satisfy \eqref{eq:rob_cbf} is given by
\begin{align} 
   \label{eq:rcbf-contr} 
    & K_\text{R-CBF}(t, \hat{\mb x}) \triangleq \big\{ \mb u \in \mathcal{U}  \big |  \dot{h}(t, \hat{\mb x}, \mb u) + \alpha (h(t, \hat{\mb x})) \\
   & \ \ \geq  \gamma_{1}(\hat{\mb x}) \left\|L_{\mb g} h(t, \hat{\mb x})\right\|_{2} + \gamma_{2}^2 (\hat{\mb x}) \left\|L_{\mb g} h(t, \hat{\mb x})\right\|_{2}^{2} \big\}.
   \nonumber
\end{align}

We can establish formal robust safety guarantees based on Definition~\ref{def:rcbf} with the help of the following theorem:
\begin{theorem}
\label{the:rcbf}   
Let ${h\!:\! \mathbb{R}_{\geq 0} \!\times\! {\mathcal{X}} \!\to\! \mathbb{R}}$ be an R-CBF for \eqref{eq:affine-dynamics}, on its $0$-superlevel set ${\C(t)}$ with an ${\alpha \!\in\! \mathcal{K}_{\infty}^{e}}$. 
Assume that ${\| {\mb x}  \!-\!  \hat{\mb x} \| \!\leq\! \delta(\mb x)}$, for all ${\mb x \!\in\! \mathcal{X}}$, for some ${\delta \!:\! \mathcal{X} \!\to\! \mathbb{R}_{\geq 0}}$, and that Assumption~\ref{as:regular} and \ref{as:compact} hold. 
Then, any locally Lipschitz feedback controller ${\mb k(t, \hat{\mb x})}$ with Lipschitz constant $\mathcal{L}_{\mb k}$, satisfying ${\mb k( t, \hat{\mb x}) \!\in\! K_\text{R-CBF}(t, \hat{\mb x})}$ for all ${(t, \hat{\mb x}) \!\in\! \mathbb{R}_{\geq 0} \!\times\! \hat{\mathcal{X}}  }$, renders the original set $\C(t)$ safe and asymptotically stable for \eqref{eq:affine-dynamics} if:
\begin{equation}
    \delta({\mb x}) \leq  \frac{\gamma_{1}({\mb x})}{\mathcal{L}_{\mb k}}, \ \forall \mb x \!\in\! \C(t).
    \label{eq:rcbf_ori}
\end{equation}
Otherwise, for any compact superlevel
set $\C_{\beta}(t)$, defined in Assumption~\ref{as:compact}, the same controller ensures safety and asymptotic stability of the inflated set given by
\begin{equation*}
%\label{eq:set}
   \!\!\!\Bar{\C}(t) \!\triangleq\!\! \left\{  \!\mb x \!\in\! \mathcal{X} \!:\! h( t, \!\mb x) \!\geq\! \alpha^{-1} \! \left ( \!-\Bar{\phi}^2  \left ( t, \gamma_1(\mb x), \!\gamma_2(\mb x), \delta(\mb x) \right ) \right )  \! \right\},
   % \label{eq:rcbf_inf}
\end{equation*}
where 
\begin{equation*}
    % \label{eq:phi_fun}
    \Bar{\phi}(t, \gamma_1(\mb x), \gamma_2(\mb x), \delta(\mb x)) \!=\! \frac{ \sigma_\beta(t, \delta(\mb x))  - \gamma_{1}({\mb x}) }{2 \gamma_{2}({\mb x}) },
\end{equation*}
\begin{equation}
\label{eq:sigma_beta}
\sigma_\beta (t, \delta(\mb x)) \!\triangleq\! \sup_{ \! \hat{\mb x} \in \C_\beta(t)} \bigg [ \sup_{\Vert \mb e \Vert \leq \delta(\mb x)} \!\big [ \Vert \mb k(t, \hat{\mb x}) \!- \mb k(t, \hat{\mb x} \!+ \mb e) \Vert \big ] \!\bigg ].
\end{equation}
\end{theorem}
\begin{proof}
\label{proof:rcbf}
The proof follows from Theorem~2 and Corollary~3 in \cite{rahal_cdc_2025}, except for the explicit time dependence of $h$, and state-dependent parameters ${\gamma_{1}, \gamma_{2}}$. The explicit time dependence of $h$ is handled by our assumption that the superlevel sets of $h$ are uniformly bounded in time.
\end{proof}
\begin{remark}
\label{re:tunable_rcbf} 
Definition~\ref{def:rcbf} and Theorem~\ref{the:rcbf} extend the R-CBF framework from \cite{rahal_cdc_2025} in two key ways: it employs state-dependent parameters ${\gamma_{1}(\hat{\mb x}), \gamma_{2}(\hat{\mb x})}$ rather than constant parameters ${\gamma_{1}, \gamma_{2} \!\in\! \mathbb{R}_{> 0}}$, and it characterizes robust safety for time-varying CBFs. Thus, this definition generalizes \textit{tunable} R-CBFs \cite{alan2021safe}, and Corollary~2 in \cite{rahal_cdc_2025} becomes a special case when the state-dependent parameters take the form ${\gamma_{1}(\hat{\mb x}) \!=\! \frac{\gamma_{1}}{\varepsilon_{1}(h(t, \hat{\mb x}))}}$ and ${\gamma_{2}^2(\hat{\mb x}) \!=\! \frac{\gamma_{2}^2}{\varepsilon_{2}(h(t, \hat{\mb x}))}}$, where ${\varepsilon_{1}, \varepsilon_{2} \!:\! \mathbb{R} \!\to\! \mathbb{R}_{> 0}}$ are continuous functions that satisfy ${\varepsilon_{i}(r) \!=\! 1}$ when ${r \!\leq\! 0}$, and ${\frac{d \varepsilon_{i}}{d r} \! > \! 0}$ when ${r \!>\! 0}$, for ${i \!=\! 1,2}$.
\end{remark}

%%%%%%%%%%%%%%%%%%%%%%%%%%%%%%%%%%%%%%%%%%%%%%%%%%%%
\subsection{Online Parameter Adaptation}
%%%%%%%%%%%%%%%%%%%%%%%%%%%%%%%%%%%%%%%%%%%%%%%%%%%%
In this subsection, an online adaptation method for the robustifying parameters in the R-CBF inequality is presented. We omit the dependence on $\mb x$ from ${\gamma_1, \gamma_2}$ and ${\delta}$ for clarity in our notation.

Recall that a controller satisfying \eqref{eq:rcbf-contr} guarantees that the following set will be forward invariant and asymptotically stable:
\begin{equation*}
    \tilde{\C}(t)  \!\triangleq\! \left\{ \mb x \!\in\! \mathcal{X} : h(t, \mb x) \geq \alpha^{-1} \! \left ( \!-\phi^2  \left ( t, \gamma_1, \!\gamma_2, \delta \right ) \right ) \right \},
\end{equation*}
where ${\phi \!:\! \R_{\geq0} \!\times\! \R_{>0} \!\times\! \R_{>0} \!\times\! \R_{\geq0} \!\to\! \R_{\geq0}}$ is a measure of the set inflation, given by
\begin{equation}
    \phi(t, \gamma_1, \gamma_2, \delta) = \frac{\max\{\sigma_\beta(t, \delta) - \gamma_1, 0\}}{2 \gamma_2}. \label{eq:inflation-func}
\end{equation}
Note that the $\max$ function in \eqref{eq:inflation-func} combines the two cases discussed in Theorem~\ref{the:rcbf}. The function $\sigma_\beta$, which is defined in \eqref{eq:sigma_beta}, with some superlevel set of $h$, $\C_\beta$, upper bounds the maximum possible change in the feedback input for a given bound on the state estimation error. We also note that this function $\sigma_\beta$ will be a function of $\gamma_1, \gamma_2$, since the closed-loop controller $\mb k$ in \eqref{eq:sigma_beta} will be a function of $\gamma_1, \gamma_2$.

Equation~\eqref{eq:inflation-func} gives an explicit expression about the safe set inflation for an uncertainty bound $\delta$ and the robustifying parameters $\gamma_1$ and $\gamma_2$. The uncertainty bound $\delta$ information, albeit time-varying and possibly conservative, is usually available, for example, provided by VIO. In the following, we then aim to minimize this set inflation by choosing proper parameters $\gamma_1$ and $\gamma_2$.

To select the parameters online, we first introduce a local version of $\sigma_\beta (t, \delta(\mb x))$ in \eqref{eq:sigma_beta} around the current state estimate $\hat{\mb x}$, since CBF conditions are evaluated pointwise at the current state estimate:
\begin{equation}
    \tilde{\sigma} ( t, \hat{\mb x}, \delta, \gamma_1, \gamma_2) = \sup_{\Vert \mb e \Vert \leq \delta} \left( \Vert \mb k(t, \hat{\mb x}) - \mb k (t, \hat{\mb x} + \mb e) \Vert \right),
    \label{eq:local_sigma}
\end{equation}
which upper bounds the variation of the feedback input due to state estimation error locally at the given state estimate. It is still not possible to compute \eqref{eq:local_sigma} explicitly, as it requires evaluating the supremum around the current estimate. We now approximate this by first sampling ${N \!\in\! \mathbb{N}}$ points, $\{\tilde{\mb x}_1, \dots, \tilde{\mb x}_N \}$, from the closed ball of radius ${\delta}$ centered at ${\hat{\mb x}}$ in $\R^n$, $B_\delta(\hat{\mb x})$, and then evaluating:
\begin{equation}
    \hat{\sigma} (t, \hat{\mb x}, \delta, \gamma_1, \gamma_2) = \max_{\mb x \in \{\tilde{\mb x}_1, \dots, \tilde{\mb x}_N \} } \Vert \mb k(t, {\mb x}) - \mb k (t, \hat{\mb x}) \Vert. \label{eq:sigma_hat}
\end{equation}
Now, at any given $\hat{\mb x}$ and $\delta$, we substitute $\hat{\sigma} (t, \hat{\mb x}, \delta, \gamma_1, \gamma_2)$ as an approximation of $\sigma_\beta (t, \delta(\mb x))$ in \eqref{eq:inflation-func}, and select the optimal constants ${\gamma_1, \gamma_2}$ by solving:
\begin{align}
&\Bar{\gamma}_1(t, \hat{\mb x}, \delta), \Bar{\gamma}_2(t, \hat{\mb x}, \delta) \nonumber  \\  &= \argmin_{\gamma_1, \gamma_2 \in \R_{> 0}}  \left( \frac{\max\{\hat{\sigma} (t, \hat{\mb x}, \delta, \gamma_1, \gamma_2) - \gamma_1, 0\}}{2 \gamma_2} \right). \label{eq:zoopt}
\end{align}
Since it is complex to explicitly compute $\frac{d \hat{\sigma}}{\gamma_1}, \frac{d \hat{\sigma}}{\gamma_2}$, we resort to using derivative-free optimization methods. We remark that $\hat{\sigma}$ in \eqref{eq:sigma_hat} is a practical, sampling-based approximation of $\tilde{\sigma}$. 

\begin{remark}
\label{re:ada_k} 
To obtain a robust safe control $\mb k(t, {\mb x})$, we may replace the nominal CBF condition in the CBF-QP with an R-CBF condition to obtain an R-CBF-QP. When optimizing the parameters $\gamma_1$ and $\gamma_2$ via \eqref{eq:zoopt}, this R-CBF-QP needs to be solved for each sampling point $ \tilde{\mb x}_i$. When  ${\mathcal{U} \!=\! \R^m}$, closed-form solutions for a QP with a single constraint can facilitate fast computation.

If we utilize multiple CBFs, resulting from multiple independent safety specifications, directly in the same QP, then we need to optimize two different robustness parameters for each CBF, which consequently increases the dimensionality of the parameter optimization problem. In contrast, by using a single unified CBF constraint in the R-CBF-QP, we only need to optimize a single set of parameters ${\Bar{\gamma}_1, \Bar{\gamma}_2}$, thereby simplifying the optimization process. This motivates the need for a general CBF unification method, as will be described in the next section. 
\end{remark}

%%%%%%%%%%%%%%%%%%%%%%%%%%%%%%%%%%%%%%%%%%%%%%%%%%%%
\subsection{Unifying CBFs via Poisson Safety Functions}
\label{sec:Poisson}
%%%%%%%%%%%%%%%%%%%%%%%%%%%%%%%%%%%%%%%%%%%%%%%%%%%%
A safety function ${h_0 \!:\! \mb y \!\mapsto\! h_0(\mb y) \!\in\! \mathbb{R}}$ is a continuous function that indicates the value of safety for a particular output coordinate ${\mb y \!\in\! \R^p}$. Crucially, a coordinate is considered safe when ${h_0(\mb y) \!\geq\! 0}$ and unsafe when ${ h_0(\mb y) \!<\! 0}$. A safety function $h_0$ differs from a CBF $h$ in that $h_0$ characterizes the desired output-space safe set without considering its relationship with the underlying system dynamics. To generate $h_0$, we employ the Poisson safety function method introduced in \cite{bahati2025poisson, bena2025geometry}. 

First, we must construct a bounded domain $\Omega$ corresponding to the interior of some desired safe set. For the case of multiple safety functions $h_i$ (these can be CBFs) associated with 0-superlevel sets $\C_i$ for ${i \!\in\! \{1,\ldots, N_h\}}$, the new desired safe set can be defined as the intersection of the 0-superlevel sets of these functions: ${\C \!=\! \bigcap_{i =1}^{N_h} \C_i}$, since the safety task is to satisfy all constraints simultaneously. Equivalently, this intersection can be expressed via the $\min$ function \cite{glotfelter2017}:
\begin{equation*}
\Omega \triangleq \bigcap_{i =1}^{N_h}  \C_i = \Big\{ \mb y \in \R^p  : \min_{i \in \{1,\ldots, N_h\}} \left [ h_i( \mb y) \right ] \geq 0 \Big\}.
\end{equation*}

Given this domain $\Omega$, we construct a unified safety function $h_0$ on $\Omega$ by solving a Dirichlet problem (boundary value problem) for Poisson's equation \cite{bahati2025poisson}:
\begin{equation}
\label{eq: Poisson}
    \Bigg\{
    \begin{aligned}
        \Delta h_0(\mb y) &= \mathcal{F}(\mb y), \quad&\forall \mb y\in\Omega, \\
        h_0(\mb y) &= 0, \quad &\forall \mb y\in\partial\Omega,
    \end{aligned}
\end{equation}
where $\Delta = \frac{\partial}{\partial\mb y} \cdot \frac{\partial}{\partial\mb y}$ is the Laplacian operator, ${\mathcal{F} \!:\! \Omega \!\to\! \R_{<0}}$ is a forcing function that yields ${h_0(\mb y) \!\geq\! 0}$ in safe regions, and the boundary $\partial\Omega$ corresponds to the desired zero level set. Similarly, we can derive $h_0$ in unsafe regions $\hat{\Omega}$ by modifying the forcing function such that ${\mathcal{F}: \hat{\Omega} \rightarrow \R_{>0}}$. 

The resultant Poisson safety function $h_0$ can be used directly as a CBF for systems whose inputs have relative degree one\footnote{This is true for systems characterized by single integrator dynamics.} with respect to $h_0$. More generally, $h_0$ can be dynamically extended to generate a CBF.

%%%%%%%%%%%%%%%%%%%%%%%%%%%%%%%%%%%%%%%%%%%%%%%%%%%%
%%%%%%%%%%%%%%%%%%%%%%%%%%%%%%%%%%%%%%%%%%%%%%%%%%%%
%%%%%%%%%%%%%%%%%%%%%%%%%%%%%%%%%%%%%%%%%%%%%%%%%%%%
\section{ An Application to Safe Tracking for Tracked Robots}
\label{sec:imple}
%%%%%%%%%%%%%%%%%%%%%%%%%%%%%%%%%%%%%%%%%%%%%%%%%%%%
%%%%%%%%%%%%%%%%%%%%%%%%%%%%%%%%%%%%%%%%%%%%%%%%%%%%

\subsection{Safety Filter Construction}

This section experimentally demonstrates the effectiveness of the proposed safe controller on a tracked mobile robot, which is tasked to follow a reference trajectory while avoiding collisions. 

Suppose that the robot is governed by unicycle dynamics:
\begin{equation}
    \begin{bmatrix}
        \dot{x} \\
        \dot{y} \\
        \dot{\theta}
    \end{bmatrix} = \begin{bmatrix}
        \cos{\theta} & 0\\
        \sin{\theta} & 0\\
        0 & 1
    \end{bmatrix}  \begin{bmatrix}
        v \\
        \omega
    \end{bmatrix} ,
\label{eq:unicycle}
\end{equation}
where ${\mb x \!=\! [x~y~\theta]^\top \!\in\! \mathcal{X} \!=\! \mathbb{R}^2 \!\times\! \mathbb{S}^1}$ represents the planar positions and heading angle with respect to an inertial frame, and ${ \mb u \!=\! [v~\omega]^\top \!\in\! \mathbb{R}^2}$ denotes the linear and angular velocities as control inputs.

We define a reference-tracking control law:
\begin{equation}
\label{eq:baseline}
    \mb {k_d}(t, \hat{\mb x}) = 
    \begin{bmatrix}
        K_v \!~ \| \mb{r_e}(t, \hat{\mb x}) \!~ \| \\
        K_{\omega} \!\!~ \left(\theta_d(t, \hat{\mb x}) - \theta \right)
    \end{bmatrix} ,
\end{equation}
where ${K_v, K_{\omega} \!\geq\! 0}$ are controller gains for the linear and angular velocities, ${\theta_d \!:\! \mathbb{R}_{\geq 0} \!\times\! \hat{\mathcal{X}} \!\to\! \mathbb{S}^1}$ is the desired robot heading, and ${\mb{r_e}(t, \mb x) \!\triangleq\! [x_d(t)~ y_d(t)]^\top \!-\! [x~ y]^\top}$ is the position-error vector with respect to the desired position ${(x_d, y_d)}$ in $\mathbb{R}^2$. The desired heading angle is given by
\begin{equation}
\label{eq:ref_theta}
\theta_d(t, \mb x) = \atan \left(y_d(t) - y,\!~ x_d(t) - x \right).
\end{equation}

For our particular safety-critical scenario (as pictured in Fig.~\ref{fig: env}), we consider a reference trajectory parameterized by:
\begin{equation}
\label{eq:ref_xy}
x_d(t) = v_{\text{ref}} t, \qquad
y_d(t) = a_d \sin(c_d x_d(t) + \varphi_d),
\end{equation}
with $a_d$, $c_d$, and $\varphi_d$ being the amplitude, spatial frequency, and phase shift, respectively, and $v_{\text{ref}}$ is the desired linear velocity in the $x$-axis. 

Next, we construct a safe set such that the mobile robot must remain within a safe rectangular region defined by ${ x_\mathrm{min} \leq x \!\leq\! x_\mathrm{max}}$ and ${y_\mathrm{min} \!\leq\! y \leq y_\mathrm{max}}$ while avoiding two circular obstacles centered at ${(x_{o,i}, y_{o,i}) \!\in\! \mathbb{R}^2}$ with radius $R_i$, ${i \!\in\! \{1,2\}}$. As such, the desired unified safe set is the intersection of the 0-superlevel sets of the following functions:
\begin{align*}
&h_i(x,y) =  \left \| [x, y]^\top - [x_{o,i}, y_{o,i}]^\top \right \| - R_i, \quad i \!\in\! \{1,2\}, \\
&h_3(x) = x - x_\mathrm{min}, \quad h_4(x) = x_\mathrm{max} - x, \\
&h_5(y) = y - y_\mathrm{min}, \quad h_6(y) = y_\mathrm{max} - y.
\end{align*}

\begin{figure}
    \centering    
    % \vskip  2.2mm
    \includegraphics[width=\linewidth]{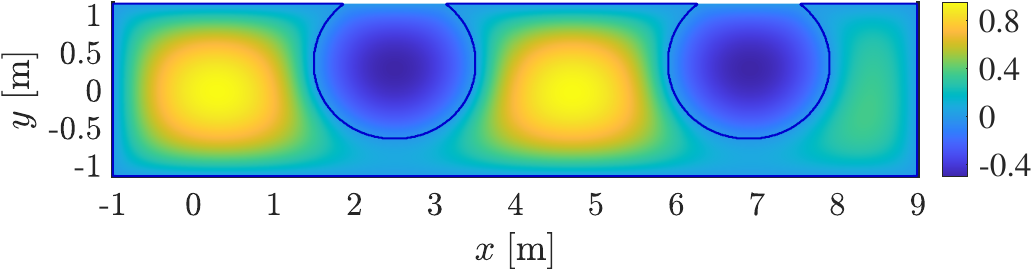}
    \caption{Poisson safety function. The function $h_0$ describes the spatial safety specification, and it can be used to generate a CBF, as in \eqref{eq:modified_CBF}.}
    \label{fig: poisson}
    % \vskip - 4mm
\end{figure} 

We generate $h_0$, depicted in \mbox{Fig.\, \ref{fig: poisson}}, by numerically solving \eqref{eq: Poisson}. Because the safe set is static, this solution can be precomputed offline on a uniform discrete mesh. During real-time operation, we evaluate $h_0$ using bilinear interpolation, and we compute its spatial derivatives with various numerical gradient methods. 

Considering the dynamics in \eqref{eq:unicycle}, the Poisson safety function $h_0$ has relative degree one with respect to the input $v$ and relative degree two with respect to the input $\omega$. This dual relative degree property reduces the utility of employing $h_0$ directly as a CBF, since the resultant safety filter will only use $v$ to enforce forward invariance of the desired safe set, completely neglecting the robot's steering capabilities. In accordance with \cite{bahati2025control}, a modified CBF is proposed as
\begin{align} \label{eq:modified_CBF}
    h(t, \mb x) &= h_0(x,y) - \frac{1}{\mu}V(t, \mb x),\\
    V(t, \mb x) &= 1 - \cos(\theta - \theta_{s}(t,x,y)),
\end{align}
where ${\mu \!\in\! \R_{>0}}$ is a tuning parameter and ${V \!:\! \mathbb{R}_{\geq 0} \!\times\! \mathcal{X} \!\to\! \mathbb{R}_{\geq 0}}$ is a Lyapunov function that characterizes the convergence of $\theta$ to $\theta_{s}(t,x,y)$ under control (\ref{eq:baseline}). The function ${\theta_{s} \!:\! \mathbb{R}_{\geq 0} \!\times\! \mathbb{R}^2 \!\to\! \mathbb{S}^1}$ defines a heading which, when achieved, guarantees safety for the unicycle system. Systematically, $\theta_{s}(t,x,y)$ can be chosen as the direction of the safe output vector for an equivalent integrator system. This modified barrier function \eqref{eq:modified_CBF} has been shown to be effective in eliminating deadlocks when avoiding obstacles in the configuration space \cite[Example~6]{cohen2024safety} and \cite{bahati2025control}. % Further details on the construction of $\theta_s$ can be found in \cite[Section~IV]{bahati2025control}.
Further details on the construction of $\theta_s$ are provided in Appendix~\ref{app:thetas}.
% Further details on the construction of $\theta_s$ are provided in the appendix.
It is also worth noting that the $\lambda$-superlevel set of $h$ will be contained within the $\lambda$-superlevel set of $h_0$ for all $t \geq 0$, for any $\lambda \in \R$, thus preserving compactness of superlevel sets.

In an R-CBF-QP, the nominal controller $\mb{k_d}( t, \hat{\mb x})$ provided to the QP is directly taken from the reference-tracking controller, $\mathbf{k_d}$. However, near the safety boundary, this nominal input often deviated from the safe direction indicated by the single-integrator velocity $\hat{v}_s$. This misalignment can lead to jittery behavior when approaching the boundary. To address this issue, we introduce a practical modification for the baseline control law $\mb {k_d}$, given in \eqref{eq:baseline}, as
\begin{equation}
\label{eq:safe_ref}
\mb{\Bar{k}_d}( t, \hat{\mb x}) =
\begin{bmatrix}
|\hat{v}_s(t, \hat{\mb x})| \\
K_{\omega} \!~ (\theta_s(t, \hat{\mb x}) -\theta)
\end{bmatrix},
\end{equation}
which aims to maintain the desired safe linear velocity while aligning the heading angle towards $\theta_s$. Practically, this modification projects the potentially unsafe reference trajectory ${[x_d(t) ~ y_d(t)]^\top}$ onto the safe set, so that the R-CBF-QP only needs to apply small corrections near the boundary, thereby mitigating the jittery behavior.

Finally, robust safety can be ensured by solving the following R-CBF-QP:
\begin{align*}
%{\label{eq:rcbf-qp}}
\begin{array}{l}
{\mathbf{k^*}(t, \hat{\mb x})= \ }
\displaystyle  \argmin_{{\mb u}  \in \mathcal{U}} \ \ \ {\|\mb u - \mb{\Bar{k}_d}( t, \hat{\mb x}) \|^2}  \\ [4mm]
~~~~~\textrm{s.t.} ~~L_{\mb f} h(t, \hat{\mb x}) + L_{\mb g} h(t, \hat{\mb x}) \mb u + \partial_t h(t, \hat{\mb x}) + \alpha (h(t, \hat{\mb x})) \\
    ~~~~~~~~~~ \geq  \Bar{\gamma}_{1}(\hat{\mb x}) \left\|L_{\mb g} h(t, \hat{\mb x})\right\|_{2} + \Bar{\gamma}_{2}^2 (\hat{\mb x}) \left\|L_{\mb g} h(t, \hat{\mb x})\right\|_{2}^{2} .
\end{array}
\end{align*}

\subsection{Simulations}
\label{sec:sim}
In this subsection, we numerically demonstrate the effectiveness of our method by comparing it against several benchmarks in both nominal (no uncertainty) and robust (with state uncertainty) scenarios. To establish a baseline for comparison, we generate a safe trajectory by solving a robust constrained optimal control problem (R-COCP)-based safety filter problem that enforces the safety and input constraints under bounded state estimation errors. For details on this R-COCP, see Appendix~\ref{sec:RCOCP}.

\begin{figure}[t]
    \centering
    %\vspace{0.22cm}
\includegraphics[width=1.0\linewidth]{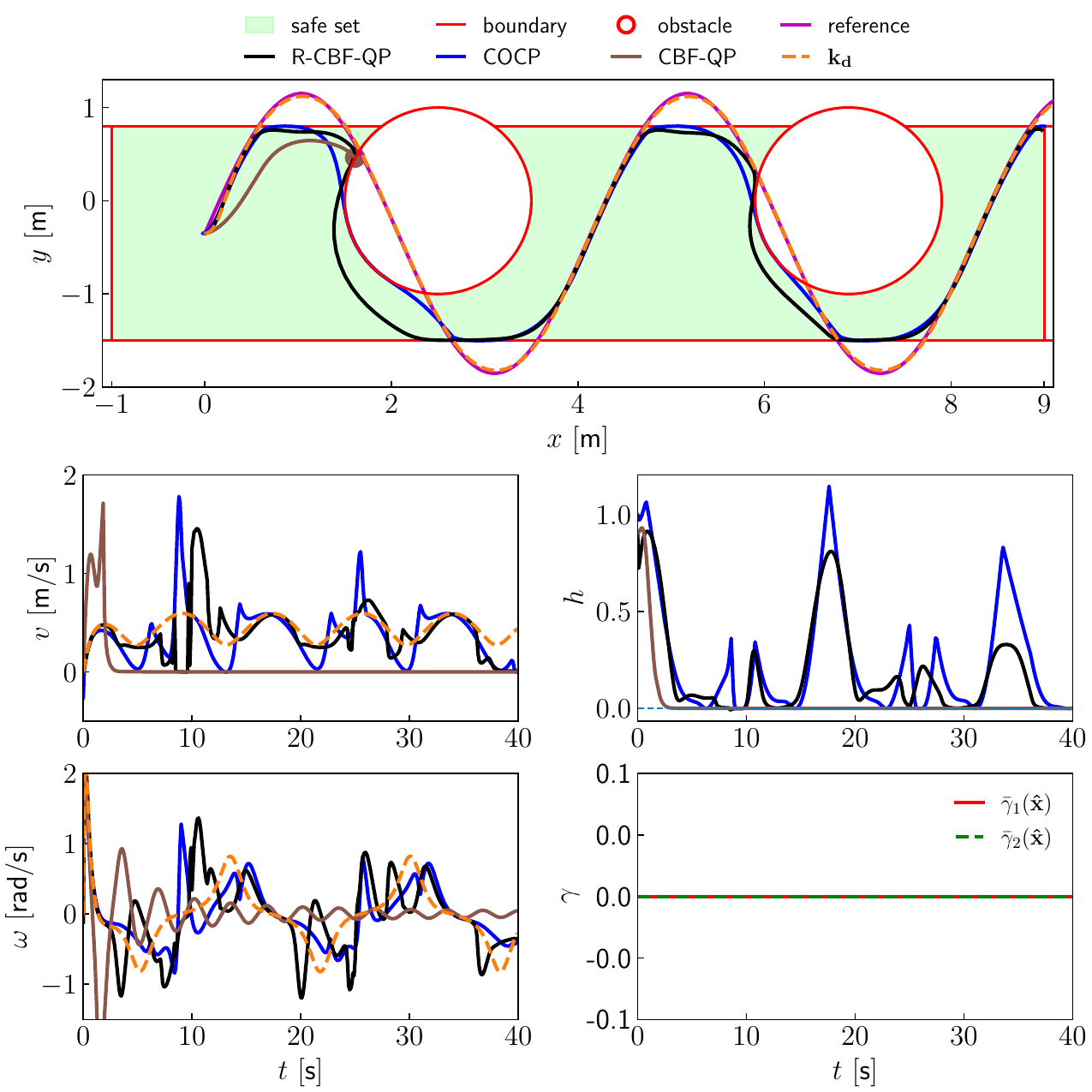}
    \caption{Nominal (no state uncertainty): COCP, vanilla CBF-QP, and our R-CBF-QP with online adaptation.  \textbf{Top:} Robot trajectories with the safe set and the reference. The brown dot ({\color[rgb]{0.549,0.337,0.294}$\bullet$}) marks a vanilla CBF-QP deadlock: the robot reaches a point where no admissible input both preserves safety and moves toward the reference. The baseline controller $\mathbf{k_d}$ follows the sinusoidal reference with negligible steady-state error. \textbf{Bottom:} Time histories of $v$, $\omega$, ${\gamma_1, \gamma_2}$, and $h$. Our method maintains ${h \!\geq\! 0}$ and closely tracks the reference, similar to COCP. The tuned parameters are obtained as ${\Bar{\gamma}_1(\hat{\mb x}) \!\approx\! 0}$, ${\Bar{\gamma}_2(\hat{\mb x}) \!\approx\! 0}$ with the proposed online tuning method, as ${\mb e \!\equiv\! 0}$.}
   \label{fig:ocp_cbf_nom}
   %\vspace{-6 mm}
	% \vskip - 4mm
\end{figure}

\begin{figure}[t]
    \centering
    %\vspace{0.22cm}
\includegraphics[width=1.0\linewidth]{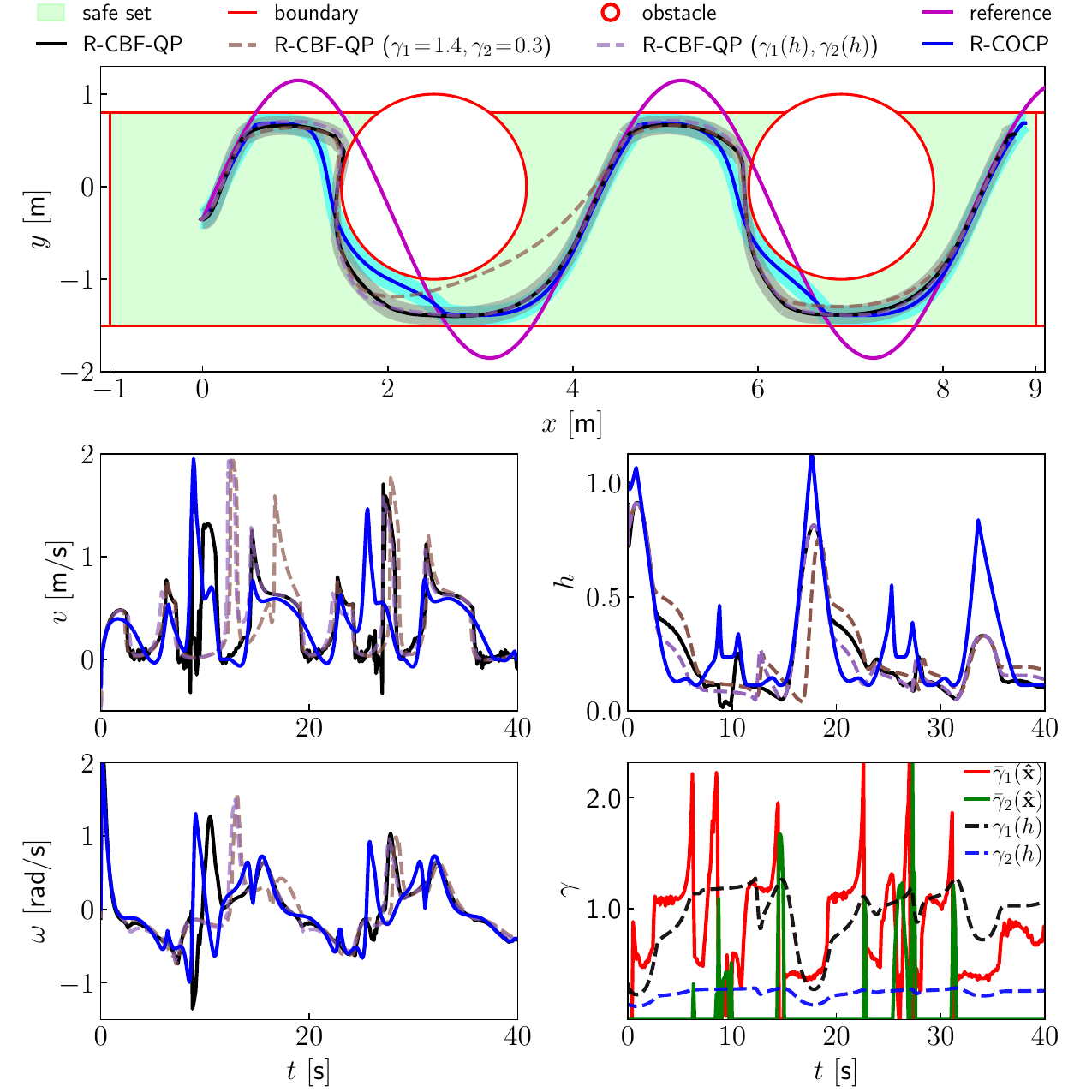}
    \caption{Robust (bounded state uncertainty): R-COCP, R-CBF-QP with fixed ${(\gamma_1 \!=\! 1.4, \gamma_2 \!=\! 0.3)}$, and tunable R-CBF-QP with ${\gamma_{1}(h), \gamma_{2}(h)}$ given in \eqref{eq:tun_par}, and our R-CBF-QP with online adaptation. \textbf{Top:} Trajectories under the measurement uncertainty. Shaded tubes show the uncertainty envelope from bounded measurement error. Our R-CBF-QP and tunable R-CBF-QP track the reference more closely and are comparable to R-COCP near obstacles, while R-CBF-QP with fixed parameters is more conservative. \textbf{Bottom:} Time histories of $v$, $\omega$, ${\gamma_1, \gamma_2}$, and $h$. All methods satisfy ${h \!\geq\! 0}$. Our online tuning method increases ${\Bar{\gamma}_1, \Bar{\gamma}_2}$ only when needed, such as near the boundary.}
   \label{fig:ocp_cbf_rob}
   %\vspace{-6 mm}
    % \vskip - 4mm
\end{figure}

In the nominal case, the benchmarks are CBF-QP and COCP. For the robust case, we compare the proposed R-CBF-QP with online robustness parameter adaptation ${\Bar{\gamma}_1(\hat{\mb{x}}), \Bar{\gamma}_2(\hat{\mb{x}})}$,  R-CBF-QP with fixed ${\gamma_1, \gamma_2}$ parameters, R-CBF-QP with tunable (CBF-dependent) ${\gamma_1(h), \gamma_2(h)}$ parameters, and R-COCP. We use two performance metrics $J_{opt}$ and $J_{t}$ (see Appendix~\ref{sec:RCOCP}) to compare alternative methods. \footnote{\label{footnote:code} The code and additional details of our simulations can be found at \\ \url{https://github.com/ersindas/robust-CBFs}} 

To simulate our safety-critical scenario with the real-world setup shown in Fig.~\ref{fig: env}, we first measure the boundaries of the safe rectangular region as ${x_\mathrm{min} \!=\! -1.5}$ m, ${x_\mathrm{max} \!=\! 0.8}$ m, ${y_\mathrm{min} \!=\! -2.0}$ m, ${y_\mathrm{max} \!=\! 2.0}$ m. And, two circular obstacles are centered at ${(x_{o,1}, y_{o,1}) \!=\! (2.5, 0.0)}$ m and ${(x_{o,2}, y_{o,2}) \!=\! (6.9, 0.0)}$ m with radius ${0.6}$ m. We inflate the circular obstacle boundaries by half the robot’s length, which is $0.4$ m; therefore, we have ${R_1 \!=\! R_2 \!=\! 1}$ m. Using these geometric parameters from the experimental setup, we define the individual CBFs ${h_i}$ for ${i \!\in\! \{1,\ldots, 6\}}$, construct the unified Poisson safety function $h_0$ as explained in Section~\ref{sec:Poisson}, and obtain $h_0$ shown in Fig.~\ref{fig: poisson}.

Next, we choose a sinusoidal reference trajectory tracking task via \eqref{eq:ref_xy} and \eqref{eq:ref_theta} such that it yields a safety constraint violation under the nominal controller $\mathbf{k_d}$ in \eqref{eq:baseline}. We set the parameters as ${v_{\text{ref}} \!=\! 0.25}$ m/s, ${a_d \!=\! 1.5}$ m, ${c_d \!=\! 1.52 }$ rad/s, and ${\varphi_d \!=\! 0}$. We use $\mathbf{k_d}$ and its safe projection $\mb{\Bar{k}_d}$ in \eqref{eq:safe_ref} with the controller gains ${K_v \!=\! 1, K_{\omega} \!=\! 2.5}$. The initial condition is ${\mb x_0 \!=\! [0~-0.35~0]^\top}$. 

For CBF-QP and R-CBF-QPs, the safety-critical control parameters are: ${\alpha \!=\! 3}$, and ${\mu \!=\! 3.3}$. We pick the fixed robustness parameters for comparison as ${\gamma_1 \!=\! 1.4}$, ${\gamma_2 \!=\! 0.3}$. To modulate the fixed robustness gains in the tunable R-CBF-QP, we utilize an exponential function following\cite{alan2021safe}: 
% \begin{equation*}
$\varepsilon_{i}(h(t, \hat{\mb x}))= {\rm e}^{\left( \eta_i \max \{h(t, \hat{\mb x}),~0 \} \right)}$,
%\end{equation*}
where ${\eta_i \!>\! 0}$, and in our simulations, we set ${\eta_1 \!=\! 2, \eta_2 \!=\! 2}$; therefore, we have 
\begin{align}
\begin{split}
\label{eq:tun_par}
    \gamma_{1}(h) &= \frac{\gamma_{1}}{{\rm e}^{\left( \eta_1 \max \{h(t, \hat{\mb x}), ~0 \} \right)}},  \\ \gamma_{2}^2(h) &= \frac{\gamma_{2}^2}{{\rm e}^{\left( \eta_2 \max \{h(t, \hat{\mb x}),\!~0 \} \right)}}. 
\end{split}
\end{align}

In this study, we employ a grid search method to solve the online tuning problem \eqref{eq:zoopt}, as proof-of-concept, though other gradient-free optimizers may also be utilized. We sample ${N \!=\! 100}$ perturbed states, and perform a gradient-free grid search over ${\Bar{\gamma}_1, \Bar{\gamma}_2}$ pairs arranged on a ${400 \!\times\! 400}$ mesh covering the range ${[0.0001, 4]}$ for each axis. 

The control input ranges for the robot are ${v \!\in\! [-2, 2]}$ m/s and ${\omega \!\in\! [-2, 2]}$ rad/s. And, the control loop operates at 50 Hz. We then solved the QP-based safety filter problems with CVXPY using the OSQP solver. We solve the robust constrained optimal control problem numerically using CasADi \cite{2019casadi} software tool using IPOPT solver, warm-starting from a forward rollout obtained with constant linear velocity ${v \!=\! v_{\text{ref}}}$ and zero angular velocity. 

The simulation results for the measurement model ${\hat{\mb x} \!\equiv\! \mb x}$ without uncertainty (nominal) are presented in Fig.~\ref{fig:ocp_cbf_nom}. Next, we consider a scenario with a bounded state measurement error set (robust) given by ${\mathcal{E} \!\triangleq\![-0.05, 0.05] \!\times\! [-0.1, 0.1] \!\times\! \{0\}}$, i.e., ${{\mb x}\!\in\! \hat{\mb x}\!\oplus\!\mathcal{E}}$, where $\oplus$ denotes the Minkowski sum. The results for this scenario are shown in Fig.~\ref{fig:ocp_cbf_rob}. 

% Simulation metrics 
\begin{table}[t]
\caption{Performance metrics for simulation in the nominal (no state uncertainty) and robust (bounded uncertainty) scenarios.}
\centering
\begingroup
\footnotesize
\setlength{\tabcolsep}{3pt}
\renewcommand{\arraystretch}{1.05}
\begin{tabular}{lcc|cc}
\toprule
& \multicolumn{2}{c}{\underline{Nominal}} & \multicolumn{2}{c}{\underline{Robust}}\\
Method & $J_{\text{opt}}$ & $J_t$ [s] & $J_{\text{opt}}$ & $J_t$ [s] \\
\midrule
COCP & 0.0 & 16.1 & -- & -- \\
CBF-QP & 22896.4 & 36.7 & -- & -- \\
R-COCP & -- & -- & 0.0 & 21.0 \\
\midrule
\textbf{R-CBF-QP} & \textbf{237.1} & \textbf{10.7} & \textbf{198.9} & 26.9 \\
R-CBF-QP ($\gamma_1\!=\!1.4,\gamma_2\!=\!0.3$) & -- & -- & 466.1 & 30.5 \\
R-CBF-QP ($\gamma_{1}(h),\gamma_{2}(h)$) & -- & -- & 312.5 & \textbf{25.3}  \\
\bottomrule
\end{tabular}
\endgroup
\label{tb:sim}
\vskip - 0mm
\end{table}
 
\textbf{Nominal (no state uncertainty):} The vanilla CBF-QP deadlocks when the robot reaches the first obstacle because the barrier involves input $v$ at relative degree one, but not $\omega$.  The filter cannot exploit steering for safety, highlighting the need to address the dual relative degree issue. With no uncertainty, our R-CBF-QP drives the robustness parameters to their lower bounds, realizing ${\Bar{\gamma}_1(\hat{\mb x}) \!\equiv \! 0.0001}$ and ${\Bar{\gamma}_2(\hat{\mb x}) \!\equiv\! 0.0001}$. By incorporating the Lyapunov term in the modified CBF $h$ \eqref{eq:modified_CBF}, our method explicitly integrates vehicle angular velocity into the CBF condition and resolves the deadlock. Moreover, the trajectory of the robot under our R-CBF-QP closely matches COCP while maintaining ${h \!\geq \!0}$, see Fig.~\ref{fig:ocp_cbf_nom}.

\begin{figure*}[t]
    \centering
%    %\vspace{0.22cm}
 \includegraphics[width=0.96\linewidth]{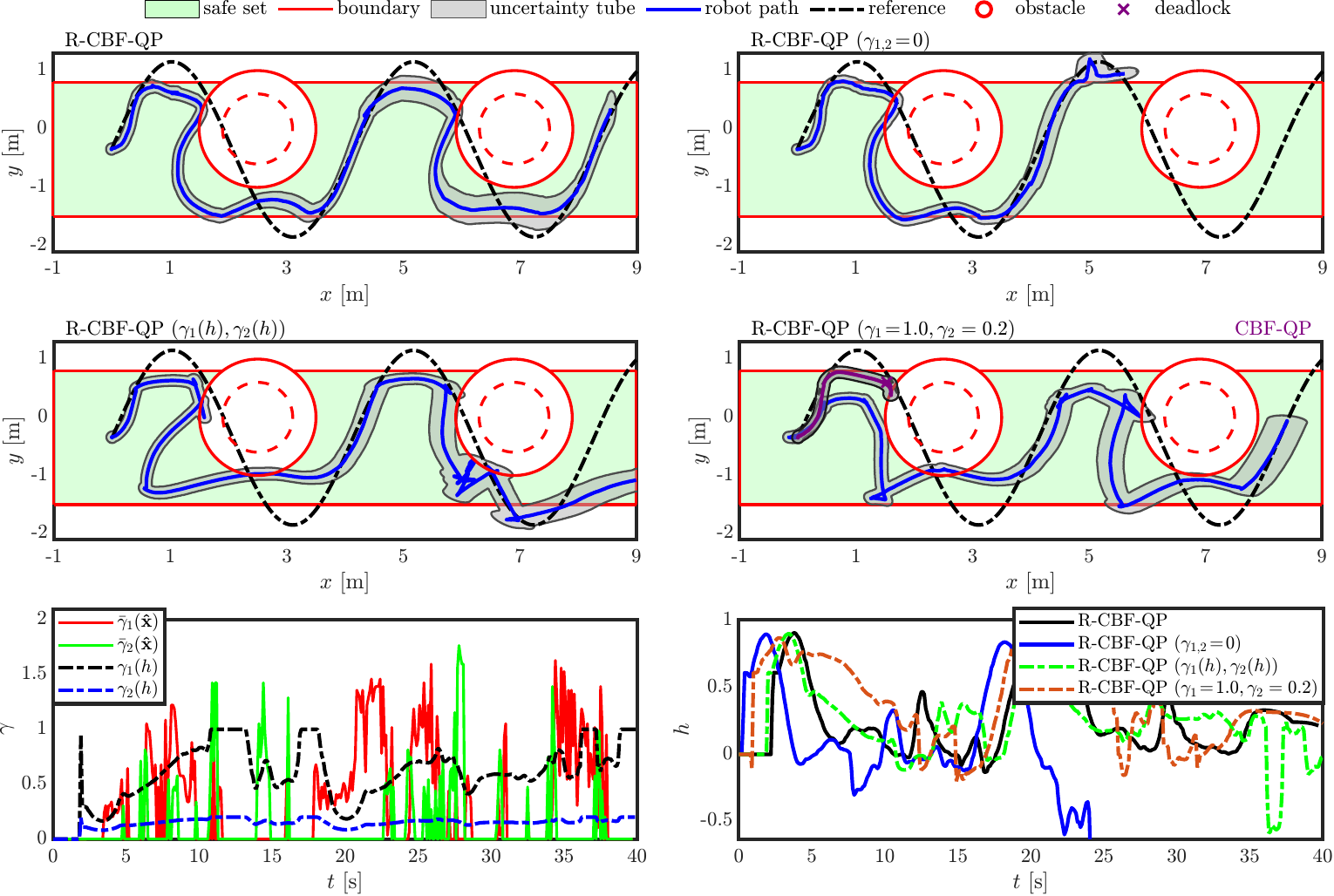}
    \caption{Experimental comparison on tracked robot. Pose is measured at the robot’s geometric center; therefore, each circular obstacle boundary is inflated by half the robot’s length. Dashed red circles denote the true safe-set boundaries, while solid red circles indicate the inflated boundaries used in the controller synthesis. Gray tubes visualize the 95\% confidence envelope of the $(x, y)$ covariance along the executed paths. \textbf{(Left)} (Top) R-CBF-QP with online adaptation of parameters ${\Bar{\gamma}_{1}, \Bar{\gamma}_{2}}$, the proposed method. \textbf{(Left)} (Middle) R-CBF-QP with tunable robustness parameters, i.e., CBF-dependent ${\gamma_{1}(h), \gamma_{2}(h)}$. Tunable R-CBF-QP violated the safety constraint in the presence of time-varying state uncertainty, while its performance was similar to R-CBF-QP with online adaptation in simulation. \textbf{(Left)} (Bottom) Robustness parameters for the adaptive and tunable R-CBF-QPs. \textbf{(Right)} (Middle) R-CBF-QP without robustness parameters, i.e., ${\gamma_{1} \!=\! 0, \gamma_{2} \!=\! 0}$. This non-robust controller violates constraints, and the robot falls from the platform and tips over at ~24 sec. The experiment was terminated after this unsafe behavior. This result illustrates the necessity of robustification in practice. The CBF-QP–controlled robot encounters a deadlock at the boundary of the first obstacle, marked by the purple cross ({\color[rgb]{0.5,0,0.5} $\times$}), since the CBF-QP only uses $v$. And, R-CBF-QP with fixed robustness parameters, i.e., ${\gamma_{1} \!=\! 1.0, \gamma_{2} \!=\! 0.2}$, this controller is more conservative. \textbf{(Right)} (Bottom) The value of CBF $h$ versus time.
    }
   \label{fig:robot}
   %\vspace{-6 mm}
	 \vskip - 0 mm
 \end{figure*}

\textbf{Robust (bounded state uncertainty):} In the presence of bounded measurement uncertainty, all controllers maintain robust safety, but differ in performance (see Fig.~\ref{fig:ocp_cbf_rob}). The R-CBF-QP with fixed parameters is noticeably conservative, slowing earlier and deviating farther from the reference, as it uses the same robustness parameters everywhere, while the tunable R-CBF-QP and our online-adapted R-CBF-QP track more closely and are comparable to R-COCP near obstacles. Our online adaptive R-CBF-QP increases ${\Bar{\gamma}_1}$ only near a boundary to add robustness against measurement uncertainty, while minimizing its value elsewhere. The value of ${\Bar{\gamma}_2}$ is mostly zero, increases only when safe set inflation is necessary, as explained in Theorem~\ref{the:rcbf}. 

In simulation, our method achieves the smallest optimal cost, ${J_{\text{opt}}}$ with ${Q \!=\!\mathrm{I}}$, ${R \!=\! \mathrm{I}}$, among the CBF-based controllers while also maintaining a competitive deviation period ${J_t}$. In contrast, the tunable R-CBF-QP method yields the lowest deviation period  ${J_t}$, but at the expense of a higher optimal cost ${J_{\text{opt}}}$. The R-COCP method provides a lower bound for this comparison. For further details, please refer to Table~\ref{tb:sim}.

\subsection{Experimental Results and Comparisons}
\label{sec:exp_com}

% Hardware metrics 
\begin{table}[t]
\caption{Performance metric $J_t$ for hardware demonstrations.}
\centering
\begingroup
\footnotesize
\setlength{\tabcolsep}{4pt}
\renewcommand{\arraystretch}{1.05}
\begin{tabular}{lc}
\toprule
Method & $J_t$ [s] \\
\midrule
\textbf{R-CBF-QP } & \textbf{16.6} \\
R-CBF-QP ($\gamma_{1,2}\!=\!0$) & 26.0 \\
R-CBF-QP ($\gamma_{1}(h),\gamma_{2}(h)$) & 26.3 \\
R-CBF-QP ($\gamma_1\!=\!1.0,\gamma_2\!=\!0.2$) & 35.8 \\
CBF-QP  & 34.2 \\
\bottomrule
\end{tabular}
\endgroup
\label{tb:exp}
% \vskip - 4mm
\end{table}

We experimentally investigated the effectiveness of the proposed safe controller on a tracked mobile robot. A photo of the robot on an elevated concrete platform during one experiment that involved two static obstacles is shown in Fig.~\ref{fig: env}. In this safety-critical navigation scenario, the robot may fall if it crosses the blue-taped rectangle. \footnote{See video at: \url{https://youtu.be/tTNu-Etm0SU}} 

Our algorithms operate on a tracked GVR-Bot modified from the iRobot PackBot 510 platform. An NVIDIA Jetson AGX Orin runs our Python and C++ code under the ROS1 framework. The linear and angular body velocity commands are sent from the AGX Orin to an onboard Intel Atom, which converts these commands into individual track speeds, which are regulated through high-rate controllers. 

The robot's perception system uses an Intel RealSense depth camera operating at 30 Hz. A VectorNav VN-100 provides inertial measurements. Our state estimation employs OpenVINS \cite{genevaOpenVINS}, a visual-inertial odometry (VIO) that fuses IMU signals and sparse visual features via an extended Kalman filter. OpenVINS uses first-estimate Jacobians to maintain observability consistency, propagating and updating the full state covariance. The covariance estimates published by OpenVINS provide us with a runtime uncertainty bound. We use a 95\% confidence interval.

We evaluated the performance of four controllers in the scenario shown in Fig. \ref{fig:robot}, in which the safe set, reference path, and obstacle geometry are the same as in simulation. The state uncertainty bound, provided by VIO, is now time–varying. The hardware study used the same hyperparameters as in Section~\ref{sec:sim} except that we employed a coarser $80\times 80$ grid for the tunable parameters and set $\alpha = 1$. On our platform, solving the R-CBF-QP with this grid takes approximately $13~\mathrm{ms}$ per control step, within the $20~\mathrm{ms}$ budget of the $50~\mathrm{Hz}$ control loop. Our online–adapted R-CBF-QP with ${\Bar{\gamma}_{1}, \Bar{\gamma}_{2}}$, consistently preserved safety while staying closest to the reference path. The fixed–gain R-CBF-QP exhibited the expected conservatism (earlier braking and a larger lateral tracking error), while the tunable R-CBF-QP became unsafe in hardware under time-varying state estimation uncertainty. Note that the tunable R-CBF-QP uses the CBF-dependent gains $\gamma_{1}(h),\gamma_{2}(h)$ given in \eqref{eq:tun_par}, which do not adapt to the time-varying $\delta(\hat{\mb x}(t))$; this can leave insufficient robustness on hardware and lead to safety violations. The R-CBF-QP without robustness parameters, i.e., ${\gamma_{1} \!=\! 0, \gamma_{2} \!=\! 0}$ also violated safety and was terminated earlier, see Fig.~\ref{fig:robot}. Quantitative outcomes are presented in Table~\ref{tb:exp}: our online adaptive R-CBF-QP achieves the lowest $J_t$ among safe controllers.

\section{Conclusion} 
\label{sec:conc}
We presented a method to synthesize robust controllers with formal safety guarantees for nonlinear systems that experience state estimation or measurement uncertainties. We tested its performance in a safe trajectory tracking task for a tracked mobile robot. To handle multiple obstacles, we incorporated a unified numerical CBF derived from Poisson’s equation. We further addressed the dual relative degree issue that typically causes deadlocks in path tracking. Experimental trials demonstrated the overall performance improvement of our approach over existing formulations. 

\begin{spacing}{0.95}
\bibliographystyle{IEEEtran}
%\section*{References}
%\vspace{-2 mm}
\bibliography{Bib/refs}
%\vspace{-3 mm}
\end{spacing}

% \vskip - 4mm
% \clearpage
\section*{Appendix}
\appendices
% \label{sec:app}
\subsection{Robust Constrained Optimal Control}
\label{sec:RCOCP}
To consider the impact of state errors, we pre-compute tightened safety constraints based on a tube model predictive control (MPC) concept \cite{mayne2005robust}. Specifically, we enlarge each circular obstacle by the maximum expected position error, while the edges of the rectangular boundary are shifted inward by the maximum anticipated error in the corresponding direction. The resulting set is referred to as the tightened set, that is, the 0-superlevel set of ${\Bar{h}_i(\mb x)}$, which is defined as the tightened version of our CBFs ${h_i(\mb x)}$ for ${i \!\in\! \{1,\ldots, 6\}}$. Therefore, when ${\Bar{h}_i(\mb x)}\geq 0$, the robot remains within the safe set, even in the presence of estimation errors. Then, we define a robust constrained optimal control problem (R-COCP): 
\begin{align*}
%\label{eq:opt_w_cbf1}
\begin{array}{l}
{{{\mb x}^*(t)}, {\mb u^*(t)} \!=\! \ }
\displaystyle  \argmin_{ \mb{x}(\cdot), \ \mb u(\cdot) } \ \ {\int_0^\infty  \big ( \|\mb u - \mb{k_d}( t, {\mb x}) \|^2} \big ) dt  \\ [4mm]
~~~~~~~\textrm{s.t.} ~~~\dot{\mb x} = \mb{f}(\mb x) + \mb{g}(\mb x) \mb u , \ \mb x(0) \!=\!  \hat{\mb x}(0), \\
~~~~~~~~~~~~~\Bar{h}_i(\mb x(t)) \!\geq\! 0,  \forall t \!\geq\! 0, i \!\in\! \{1,\ldots, 6\}, \\
~~~~~~~~~~~~~ \mb u(t)\in\mathcal{U},
\end{array}
\end{align*}
where ${\mb u^*(t)}$ and ${\mb x^*(t)}$ are the respective optimal input and state signals.

In order to quantify the optimality of trajectories with respect to the constrained optimal control problem, we use a standard cost metric:
\begin{equation*}
   J_{opt} \!\triangleq\! \int_0^\infty  ({\mb x} - {\mb x}^*)^\top Q~\!({\mb x} - {\mb x}^*) + (\mb u - \mb u^*)^\top R~\!(\mb u - \mb u^*) dt ,
\end{equation*}
where the first term penalizes deviations of the states from the optimal trajectory, while the second term penalizes deviations of the control input from the optimal control input. In this context, ${({\mb x}^*, \mb u^*)}$ is obtained from the solution to R-COCP that incorporates tightened safety constraints. 

In addition to the standard cost metric, we introduce another metric that indicates the time optimality for tracking tasks, defined as
\begin{equation*}
   J_{t} \!\triangleq\!  \int_0^\infty \mathds{1}_{\hat{\mathcal{X}}} \left ( \| \hat{\mb x}(t) - \hat{\mb x}_{\mb{k_d}}(t) \| \geq \epsilon_{\mb x} \right )  ~\! dt,
\end{equation*}
where ${\epsilon_{\mb x} \!\geq\! 0}$, ${\hat{\mb x}_{\mb{k_d}}}$ is the state under the reference-tracking controller ${\mb{k_d}}$, and ${\mathds{1}_{\hat{\mathcal{X}}}}$ is the indicator function such that ${\mathds{1}_{\hat{\mathcal{X}}} \left ( \| \hat{\mb x}(t) \!-\! \hat{\mb x}_{\mb{k_d}}(t) \| \geq \epsilon_{\mb x} \right ) \!=\! 1}$ if $ \| \hat{\mb x}(t) \!-\! \hat{\mb x}_{\mb{k_d}}(t) \| \geq \epsilon_{\mb x}$, and $0$ otherwise. Therefore, ${J_{t}}$ represents the time during which the error between the resulting states from the safe control and the states derived from ${\mb{k_d}}$ exceeds ${\epsilon_{\mb x}}$, and quantifies the practical tracking performance. For our safe control problem, this metric considers only the positional error, replacing the full state vector $\hat{\mb x}(t)$ with the robot's position ${[x(t)~y(t)]^\top}$, and the reference states $\hat{\mb x}_{\mb{k_d}}$ with the desired position ${[x_d(t)~y_d(t)]^\top}$.

\subsection{Construction of a Desired Heading Angle}
\label{app:thetas}
%\textcolor{blue}{This section is only for the arXiv submission.}

Following the approach in \cite{bahati2025control}, the desired heading angle comes from a safety filter for the single integrator dynamics: 
\begin{equation*}
    \begin{bmatrix}
        \dot{x} &
        \dot{y} 
    \end{bmatrix}^\top = \hat{v} 
\end{equation*}
with ${\hat{v} \!\triangleq\! [\hat{v}_x~\hat{v}_y]^\top \!\in\! \mathbb{R}^2}$. In order to derive $\theta_{s}(t,x,y)$, we first introduce a proportional tracking controller for the single integrator dynamics:
\begin{equation}
\label{eq:v_p}
     \hat{v}_p(t,x,y) = K_v \!~  \mb{r_e}(t, \mb x) .
\end{equation}
A safety filter using the classic CBF-based QP is given by
\begin{equation*}
\begin{aligned}
    \hat{v}_{s} (t,x,y)  & = \argmin_{\hat{v}} \| \hat{v} - \hat{v}_p(t,x,y) \| ^2  \\
    & \text{s.t.} \  \grad h_0 \cdot \hat{v} +  \alpha h_0(x,y) \geq 0.
\end{aligned}
\end{equation*}
This gives an explicit solution: 
\begin{equation*}
   \hat{v}_{s} (t,x,y)  \!=\! \hat{v}_p + \max \!\left (  -(\grad h_0 \cdot \hat{v}_p \!+\!  \alpha h_0)/\| \grad h_0\|^2,0 \right )  \grad h_0^\top. 
\end{equation*}
However, due to the $\max$ function, the above solution may not be differentiable. Instead, a smooth safety filter was proposed in \cite{cohen2023characterizing} that has an explicit form:
\begin{equation}
\label{eq:v_safe}
 \hat{v}_{s} (t,x,y) \triangleq [\hat{v}_{s,x}~\hat{v}_{s,y}]^\top = \hat{v}_p + \lambda(a,b)\grad h_0^{\top}
\end{equation}
with ${a(t,x,y) \!\triangleq\!  \grad h_0 \cdot \hat{v}_p \!+\!  \alpha h_0(x,y)}$, ${b(x,y) \!\triangleq\! \| \grad h_0\|^2}$ and function $\lambda$ that satisfies certain conditions. One specific choice of the function $\lambda$ is the half Sontag formula:
\begin{equation*}
      \lambda(a,b) =  \frac{-a + \sqrt{a^2 + q(b) \!~ b}}{2b} , 
\end{equation*}
where ${q \!:\! \mathbb{R} \!\to\! \mathbb{R}}$, ${q(0) \!=\! 0}$, and ${q(b) \!>\! 0}$ for all ${b \!\neq\! 0}$. 
We use ${q(b) \!\triangleq\! \alpha_q b}$ with a parameter ${\alpha_q \!>\! 0}$. Then, the safe heading angle is obtained as
\begin{equation}
\label{eq:theta_safe}
    \theta_{s}(t,x,y) = \atan(\hat{v}_{s,y}, \hat{v}_{s,x}).
\end{equation}

% \end{spacing}

\end{document}